\documentclass[journal]{IEEEtran}
\ifCLASSINFOpdf
\else
\fi
\usepackage[comma,numbers,square,sort&compress]{natbib}
\usepackage{algorithm} 
\usepackage{algorithmic} 
\usepackage{setspace}
\usepackage{arydshln}
\usepackage{multirow} 
\usepackage{xcolor}
\usepackage{setspace}
\usepackage{amsmath}
\newtheorem{theorem}{\textbf{Theorem}}
\newtheorem{lemma}{\textbf{Lemma}}
\newtheorem{example}{\textbf{Example}}
\newtheorem{corollary}{\textbf{Corollary}}
\newtheorem{remark}{\textbf{Remark}}
\newtheorem{definition}{\textbf{Definition}}
\newtheorem{problem}{\textbf{Problem}}
\newtheorem{proposition}{\textbf{Proposition}}
\usepackage{pmat}
\newenvironment{proof}{{\noindent{\bf \emph{Proof:}}}\quad}{\hfill $\square$\par}

\usepackage{multirow}
\usepackage{url}
\usepackage{subfigure}
\usepackage{fancyhdr}
\usepackage{amsmath}
\usepackage{multirow}
\usepackage{amssymb }
\usepackage{color}
\usepackage{graphics} 
\usepackage{graphicx}
\usepackage{geometry}
\usepackage{setspace}

\begin{document}
%
%
%
%

\title{\LARGE \bf
Structural Controllability of Undirected Diffusive Networks with Vector-Weighted Edges
}

%
%



\author{Yuan Zhang, Yuanqing Xia, Han Gao, and Guangchen, Zhang
\thanks{The authors are with the School of Automation, Beijing Institute of Technology, Beijing, China. 
       {(email: {\tt\small zhangyuan14@bit.edu.cn, xia\_yuanqing@bit.edu.cn, gaohhit@163.com, guangchen\_123@126.com}).}}}
\maketitle

\begin{abstract} In this paper, controllability of undirected networked systems with {diffusively coupled subsystems} is considered, where each subsystem is of {identically {\emph{fixed}}} general high-order single-input-multi-output dynamics. The underlying graph of the network topology is {\emph{vector-weighted}}, rather than scalar-weighted. The aim is to find conditions under which the networked system is structurally controllable, i.e., for almost all vector values for interaction links of the network topology, the corresponding system is controllable. It is proven that, the networked system is structurally controllable, if and only if each subsystem is controllable and observable, and the network topology is globally input-reachable. These conditions are further extended to the cases {with multi-input-multi-output subsystems and matrix-weighted edges,} or where both directed and undirected interaction links exist.

\end{abstract}


\section{INTRODUCTION}
Analysis and synthesize of networked systems with {diffusively coupled} subsystems, also known as {\emph{diffusive networks}} in some literature \cite{LucaSynchronization,zhang2014upper}, have received much attention in the fields of synchronization, consensus, stability, as well as controllability and observability \cite{LucaSynchronization,mesbahi2010graph,zhang2014upper}. This is because the {diffusive coupling} mechanism frequently arises naturally in thermal systems, power systems, car-following traffic systems, as well as opinion propagations in social networks (see examples in \cite{mesbahi2010graph,Modern_Control_Ogata}). As is known to all, controllability is a fundamental system property. Particularly, controllability of a leader-follower multi-agent system (MAS) running the consensus protocol guarantees that the system can reach agreement subspace arbitrarily fast \cite{R.Am2009Controllability}. 

Many works have focused on controllability of leader-follower MASs \cite{H.G.2004On,R.Am2009Controllability,aguilar2015graph,zhang2014upper}.  For example, \cite{H.G.2004On} gave necessary and sufficient conditions for controllability of such networked systems in terms of eigenvectors of Laplacian matrices. The works \cite{R.Am2009Controllability,zhang2014upper,aguilar2015graph} studied the same problem by means of graph-theoretic tools.  However, except \cite{zhang2014upper}, all the above works assume that each subsystem is a single-integrator. 

On the other hand, controllability of networked systems with high-order subsystems has also attracted much research interests in \cite{L.Wa2016Controllability,Y_Zhang_2016,trumpf2018controllability,xue2018modal,zhang2019structural,Commault2019Generic}. To be specific, \cite{L.Wa2016Controllability,xue2018modal,Commault2019Generic} focused on networked systems with identical subsystems (called homogeneous networks), whiles \cite{Y_Zhang_2016,trumpf2018controllability,zhang2019structural} on networked systems with general heterogeneous subsystems (called heterogeneous networks). {Particularly, controllability as a generic property for a networked system is studied in \cite{zhang2019structural,Commault2019Generic}.}  

However, except \cite{Y_Zhang_2016,trumpf2018controllability,zhang2019structural} which focus on heterogeneous networks, almost all results on controllability of homogeneous networks are built on the condition that all weights of edges in the network topology belong to $\{0,1\}$ or some scalars \cite{H.G.2004On,R.Am2009Controllability,aguilar2015graph,zhang2014upper,L.Wa2016Controllability,xue2018modal,Commault2019Generic}. Notice that, when each subsystem is not of single-input-single-output (SISO), there is a typical situation that different interaction channels between two subsystems are weighted differently, either because of differences in the nature of physical variables they convey, or the variants of the channels themselves. For example, in some networks consisting of both physical coupling and cyber coupling, the physical channels and the cyber ones between two subsystems can have different weights. See more examples in \cite{tuna2016synchronization,tuna2017observability,zhang2019Structural_arxiv}. If we use a graph to denote the subsystem interaction topology (i.e., network topology), then each edge of the graph { may have a {\emph{vector-valued}} or {\emph{matrix-valued}} weight as introduced in \cite{tuna2016synchronization,tuna2017observability}. In such case, some existing approaches for controllability analysis for homogeneous networks with scalar-weighted edges may not be applicable (such as the spectrum-based approaches in \cite{zhang2014upper,xue2018modal}).}  

In this paper, we study structural controllability of an undirected diffusive networked system with high-order subsystems, where the underlying graph of the network topology has symmetric vector-weighted edges.  Our purpose is to find necessary and sufficient conditions under which the networked system is structurally controllable, i.e., for almost all vector values for edges of the network topology, the whole system is controllable.  The main contributions of this paper are three-fold.  First, we prove that, an undirected diffusive networked system with identical {{\emph{single-input-multi-output (SIMO)}} subsystems (leading to vector-weighted edges in the network topology) is structurally controllable, if and only if each subsystem is controllable and observable, and the network topology is globally input-reachable. Second, we show that our conditions are still valid even when both directed edges and undirected ones exist. {Third, we extend our results to the case with matrix-weighted edges, where the weight matrices can be of arbitrary dimensions.}


{It is remarkable that some relations between connectivity and observability have been revealed in \cite{tuna2017observability} for a networked system, in which subsystems are decoupled whereas their outputs are coupled by sensor networks, and each interconnection edge defined therein has a {\emph{semi-definite weight}}. Such setting is obviously different from the class of systems studied in this paper. It is also mentionable that (strong) structural controllability of networks of single-integrators with symmetric weights (or more complicated dependencies than symmetry) has recently received much attention in \cite{mousavi2017structural,menara2018structural,li2018structural,jia2019strong}. This paper differs from these works in the sense that dependencies exist between self-loop of a vertex and the edges connecting to it.
Finally, although heterogeneous networks described in \cite{Y_Zhang_2016,trumpf2018controllability,zhang2019structural} may cover the system model studied in this paper, their results are essentially rank conditions \cite{Y_Zhang_2016,trumpf2018controllability} whose verifications usually require algebraic calculations in the global system level, or some combinatorial tools like matroid \cite{zhang2019structural}, rather than simple topological conditions herein.}}


The rest is organized as follows. Section II gives the problem formulation. Section III presents the main results, with Section IV providing the proofs. {Section V considers the case with matrix-weighted edges.}  Finally, some concluding remarks are given in Section VI.

{\emph{Notations:}}  For a set, $|\cdot |$ denotes its cardinality. For a matrix $M$, $M_{ij}$ or $[M]_{ij}$ denotes the entry in the $i$th row and $j$th column of $M$. $\sigma(M)$ denotes the set of eigenvalues of the square matrix $M$, and ${\bf diag}\{X_i|_{i=1}^n\}$ the block diagonal matrix whose $i$th diagonal block is $X_i$,  ${\bf col}\{X_i|_{i=1}^n\}$ the matrix stacked by $X_i|_{i=1}^n$. By $e_i^{[N]}$ we denote the $i$th column of the $N$ dimensional identify matrix $I_N$, and ${\bf 1}_{m\times n}$ the $m\times n$ matrix with entries all being one.
\section{Problem Formulation}\label{model_description}
Consider a networked system consisting of $N$ subsystems. Let ${\cal G}_{\rm sys}=({\cal V}_{\rm sys}, {\cal E}_{\rm sys})$ be an {\emph{undirected graph}} without self-loops describing the subsystem interconnection topology (i.e., underlying graph of the network topology), with ${\cal V}_{\rm sys}=\{1,...,N\}$, and an undirected edge $(i,j)\in {\cal E}_{\rm sys}$ if the $j$th subsystem and the $i$th one are directly influenced by each other. Dynamics of the $i$th subsystem is described by
\begin{equation} \label{sub_dynamic}
\dot x_i(t)= Ax_i(t) + bv_i(t),
\end{equation}
where $A\in {\mathbb R}^{n\times n}$, $b\in {\mathbb R}^{n}$, $x_i(t)\in {\mathbb R}^{n}$ is the state vector, $v_i(t)\in {\mathbb R}$ is the input injected to the $i$th subsystem. The input $v_i(t)$ may contain both subsystem interactions and the external control inputs, expressed as
\begin{equation} \label{sub_interaction}
v_i(t)=\delta_i u_i(t)+ \sum \nolimits_{j=1,j\ne i}^N{W_{ij}}C(x_j(t)-x_i(t)),
\end{equation}
where $u_i(t)$ is the external control input, $C\doteq[c_1^{\intercal},\cdots,c_r^{\intercal}]^{\intercal}\in{\mathbb R}^{r\times n}$ is the subsystem output matrix with $c_k\in {\mathbb R}^{1\times n}$, and $W_{ij}\in {\mathbb R}^{1\times r}$ is the vector-valued weight of edge from the $j$th subsystem to the $i$th one. Denote the $k$th element of $W_{ij}$ by $w^{[k]}_{ij}$, $k=1,...,r$, i.e.,
 $w^{[k]}_{ij}\in {\mathbb R}$ is the weight imposed on $c_k(x_j(t)-x_i(t))$. Moreover, $\delta_i\in \{0,1\}$, $\delta_i=1$ means that the $i$th subsystem is directly controlled by the external input $u_i(t)$, and $\delta_i=0$ means the contrary.  In addition, $W_{ij}=W_{ji}$ for $i,j\in\{1,...,N\}$, and $W_{ij} \ne 0$ only if $(j,i)\in {\cal E}_{\rm sys}$ (,$i\ne j$).


Let $\Delta={\bf diag}\{\delta_i|_{i=1}^N\}$, $u(t)=[u_1(t),...,u_N(t)]^{\intercal}$, $x(t)=[x_1^{\intercal}(t),...,x_N^{\intercal}(t)]^{\intercal}$. Define matrix $L_k\in {\mathbb R}^{N\times N}$ as $[L_k]_{ij}=- w_{ij}^{[k]}$ if $i \ne j$, and $[L_k]_{ij}=\sum\nolimits_{p = 1,p \ne i}^N {w_{ip}^{[k]}}$ if $i=j$, for $k=1,...,r$. Then, $L_1,...,L_r$ are (scalar-weighted) Laplacian matrices associated with the undirected graph ${\cal G}_{\rm sys}$. The lumped state-space representation of the networked system then is
\begin{equation} \label{lump_ss}
\dot x(t)=A_{\rm sys}x(t) + B_{\rm sys}u(t),
\end{equation}with
\begin{equation} \label{lump_pp} \begin{aligned}
 A_{\rm sys}&=I_N\otimes A-\sum \nolimits_{k=1}^r L_k\otimes(bc_k)\\
 &=I_N\otimes A-(I_N\otimes b) L_g (I_N\otimes C),\\ B_{\rm sys}&=\Delta \otimes b, \end{aligned}\end{equation}where $\otimes$ denotes the Kronecker product, and {\small\begin{equation}\label{vector_weight} L_g=\left[
           \begin{array}{cccc}
             \sum \nolimits_{j=1}^N W_{1j} & -W_{12} & \cdots & -W_{1N} \\
             -W_{21} & \sum \nolimits_{j=1}^N W_{2j} & \cdots & -W_{2N} \\
             \vdots & \vdots & \ddots & \vdots \\
             -W_{N1} & -W_{N2} & \cdots & \sum \nolimits_{j=1}^N W_{Nj} \\
           \end{array}
         \right]\in {\mathbb R}^{N\times Nr}
\end{equation}}is a {\emph{vector-weighted Laplacian}} \cite{tuna2016synchronization}.
Throughout this paper, without losing of generality, assume that $c_k\ne 0$, for $k=1,...,r$.

The (\ref{sub_dynamic})-(\ref{sub_interaction}) models a diffusive networked system with identical subsystems, which arises in modeling interacted liquid tanks \cite{Modern_Control_Ogata}, synchronizing networks of linear oscillators \cite{LucaSynchronization,zhang2019structural}, electrical networks \cite{tuna2017observability}, consensus-based MASs \cite{mesbahi2010graph}, etc. Specially, when $r=1$, (\ref{sub_dynamic})-(\ref{sub_interaction}) becomes a networked system with SISO subsystems. {Readers are referred to \cite{tuna2016synchronization,zhang2019Structural_arxiv} for more examples for networked systems with vector-weighted edges. }

\begin{definition} Given $A,b,C,\Delta$ and an undirected ${\cal G}_{\rm sys}$, the networked system (\ref{sub_dynamic})-(\ref{sub_interaction}) is said to be {\rm{structurally controllable}}, if there exists a set of values for $\{W_{ij}\}_{(j,i)\in {\cal E}_{\rm sys}}$ with $W_{ij}=W_{ji}$, such that the associated $(A_{\rm sys}, B_{\rm sys})$ is controllable.
\end{definition}

{In line with \cite{menara2018structural}}, it can be shown that controllability of the networked system (\ref{sub_dynamic})-(\ref{sub_interaction}) is a {\emph{generic property}} in the sense that, if this system is structurally controllable, then for almost all values for $\{W_{ij}\}_{(j,i)\in {\cal E}_{\rm sys}}$ with $W_{ij}=W_{ji}$, the corresponding system is controllable. In practise, due to parameter uncertainties or geographical distance between subsystems, the exact weights $W_{ij}$ might be hard to know. Under such circumstance, structural controllability may be a good alternative for controllability evaluation. The main problem considered in this paper is as follows.

\begin{problem} \label{prob1} Given $A,b,C,\Delta$ and an undirected subsystem interaction topology ${\cal G}_{\rm sys}$, find necessary and sufficient conditions under which
the system (\ref{sub_dynamic})-(\ref{sub_interaction}) is structurally controllable.
\end{problem}

\vspace*{-1.0em}
\section{Main Results} \label{main_result}
In this section, we first give necessary and sufficient conditions for Problem \ref{prob1}. We then extend them to the case with semi-symmetric constrained topologies. All proofs are postponed to Section \ref{proofs}.

Let ${\cal I}_u=\{i: \delta_i\ne 0\}$ be the set of indices of subsystems that are directly influenced by external inputs, and  ${\cal U}=\{u_i: i\in {\cal I}_u\}$. Let $\bar {\cal G}_{\rm sys}=({\cal V}_{\rm sys}\cup {\cal U}, {\cal E}_{\rm sys}\cup {\cal E}_{ux})$, where ${\cal E}_{ux}=\{(u_i,i), i\in {\cal I}_u\}$. Then, $\bar {\cal G}_{\rm sys}$ reflects the information flows of the networked system. A path from vertex $i_1$ to vertex $i_p$ is a sequence of edges $(i_1,i_2),(i_2,i_3),\cdots,(i_{p-1},i_p)$, where each edge is either directed or undirected.

\begin{definition} \label{definition input-reachability}
We say a vertex $i$ is input-reachable, if there exists a path beginning from any $u \in {\cal U}$ and ending at $i$ in $\bar {\cal G}_{\rm sys}$. If every vertex of $i\in {\cal V}_{\rm sys}$ is input-reachable, we just say that $\bar {\cal G}_{\rm sys}$ (or the network topology) is globally input-reachable ({i.e., global input-reachability  means that $\bar {\cal G}_{\rm sys}$ can be decomposed into a collection of disjoint spanning trees rooted at $\cal U$}).
\end{definition}

\begin{theorem} \label{theorem_simo}
Assume that $|{\cal I}_u|<N$. Then the networked system (\ref{sub_dynamic})-(\ref{sub_interaction}) is structurally controllable, if and only if 

1) $(A,b)$ is controllable, and $(A, [c_1^{\intercal},...,c_r^{\intercal}]^{\intercal})$ is observable;

2) $\bar {\cal G}_{\rm sys}$ is globally input-reachable.
\end{theorem}

In the above theorem, we have ruled out the trivial case where $|{\cal I}_u|=N$, under which the networked system is always structurally controllable whenever $(A,b)$ is controllable (which is always necessary for the networked system to be controllable \cite{Y_Zhang_2016}). The above theorem implies that, if each subsystem is controllable and observable, and the networked topology is globally input-reachable, then for almost all vector-valued weights, the corresponding networked system is controllable. 
\begin{example}
Consider the mass-spring-damper system which consists of $N$ subsystems shown in Fig. \ref{damping_vehicle} (also used in \cite{zhang2019structural,Modern_Control_Ogata}). Let $\mu_i$ and $k_i$ denote the constants of the $i$th damper and spring, respectively, and $m$ is the mass, which is {\emph{identical}} for all subsystems. Let $x_i$ be the placement of the mass. Then, dynamics of the $i$th mass can be described as\begin{equation}\label{vehicle_damp}\begin{array}{l}
{{\ddot x}_i} = m^{ - 1}{\mu _i}({{\dot x}_{i - 1}} - {{\dot x}_i}) + m^{ - 1}{k_{i + 1}}({x_{i + 1}} - {x_i})\\
{\kern 1pt} {\kern 1pt} {\kern 1pt} {\kern 1pt} {\kern 1pt} {\kern 1pt} {\kern 1pt} {\kern 1pt} {\kern 1pt} {\kern 1pt} {\kern 1pt} {\kern 1pt} {\kern 1pt} {\kern 1pt} {\kern 1pt} {\kern 1pt} {\kern 1pt} {\kern 1pt} {\kern 1pt} {\kern 1pt} {\kern 1pt} {\kern 1pt}  - m^{ - 1}{\mu _{i + 1}}({{\dot x}_i} - {{\dot x}_{i + 1}}) - m^{ - 1}{k_i}({x_i} - {x_{i - 1}}) + m^{ - 1}{u_i}
\end{array}\end{equation}
where $u_i$ is the force imposed on the $i$th mass, and $x_0\equiv 0$, $\mu_{N+1}=0$, $k_{N+1}=0$. The above equation can be rewritten as
{\tiny
\[\begin{array}{l}
\left[ \begin{array}{l}
{{\dot x}_{i1}}\\
{{\dot x}_{i2}}
\end{array} \right] = \left[ {\begin{array}{*{20}{c}}
0&1\\
0&0
\end{array}} \right]\left[ \begin{array}{l}
{x_{i1}}\\
{x_{i2}}
\end{array} \right] + \underbrace {\left[ \begin{array}{l}
0\\
1
\end{array} \right]}_{b} \left\{ [\!\!\underbrace {\frac{{{k_i}}}{{{m}}}}_{w_{i,i - 1}^{[1]}},\underbrace{\frac{{{\mu _i}}}{{{m}}}}_{w_{i,i - 1}^{[2]}}\!\!]\underbrace {\left[ {\begin{array}{*{20}{c}}
1&0\\
0&1\\
\end{array}} \right]}_{C}\left[{\begin{array}{*{20}{c}}
{{x_{i - 1,1}} - {x_{i1}}}\\
{{x_{i - 1,2}} - {x_{i2}}}
\end{array}} \right] \right. \\
\left.+ [\!\!\underbrace {\frac{{{k_{i + 1}}}}{{{m}}}}_{w_{i,i + 1}^{[1]}},\underbrace {\frac{{{\mu _{i + 1}}}}{{{m}}}}_{w_{i,i + 1}^{[2]}}\!\!]\underbrace {\left[ {\begin{array}{*{20}{c}}
1&0\\
0&1\\
\end{array}} \right]}_{C}\left[ {\begin{array}{*{20}{c}}
{{x_{i + 1,1}} - {x_{i1}}}\\
{{x_{i + 1,2}} - {x_{i2}}}
\end{array}} \right] + \frac{{{u_i}}}{{{m}}}\right\}
\end{array},\]}where $x_{i1}=x_i$, $x_{i2}=\dot x_{i}$. From the above formula, $w_{i,i-1}^{[1]}=\frac{k_i}{m}=w_{i-1,i}^{[1]}$, and $w_{i,i-1}^{[2]}=\frac{\mu_i}{m}=w_{i-1,i}^{[2]}$. Hence, the whole system is an undirected networked system with single-input-$2$-output subsystems. The underlying graph of the network topology is a chain shown in Fig. \ref{graph} with vector-weighted edges $\{[w_{i,j}^{[1]},w_{i,j}^{[2]}]\}|_{i=1,j=i\pm 1}^N$. By Theorem \ref{theorem_simo}, assuming that $\{\mu_i,k_i|_{i=1}^N\}$ are algebraically independent, we know that this system is structurally controllable by driving arbitrarily one subsystem. \hfill $\square$

\begin{figure}
  \centering
  \includegraphics[width=3.0in]{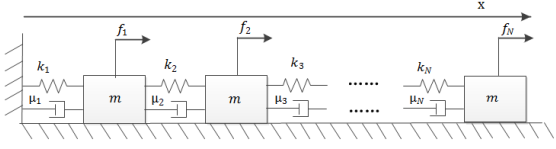}
  \caption{The mass-spring-damper system}\label{damping_vehicle} 
\end{figure}
\begin{figure}
  \centering
  \includegraphics[width=2.4in]{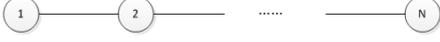}
  \caption{The underlying graph of the mass-spring-damper system}\label{graph} 
\end{figure}
\end{example}

Two direct corollaries are listed as follows.

\begin{corollary} \label{undirecte_la} Let $L$ be the weighted Laplacian matrix of a connected undirected graph $\cal G$ with $N$ vertices. Then, for almost all weights of edges of $\cal G$, $(-L,e_i^{[N]})$ is controllable, $\forall i\in\{1,...,N\}$.
\end{corollary}

\begin{corollary} \label{theorem_SISO}Suppose that in the networked system (\ref{sub_dynamic})-(\ref{sub_interaction}), every subsystem is {SISO (i.e., $r=1$)}, and assume that $|{\cal I}_u|< N$. Then this system is structurally controllable, if and only if $(A,b)$ is controllable, $(A, C)$ is observable, and $\bar {\cal G}_{\rm sys}$ is globally input-reachable.
\end{corollary}

\begin{remark}\label{remarkxx}{Structural controllability of undirected networks of single-integrators running the consensus protocol has been discussed in \cite{goldin2013weight}.} Corollary \ref{undirecte_la} differs from \cite{goldin2013weight}, as the result of \cite{goldin2013weight} is under the condition that the total sum of each row of the lumped state transition matrix $A_{\rm sys}$ and input matrix $B_{\rm sys}$ is zero, rather than that the sum of each row of $A_{\rm sys}$ is zero.
\end{remark}

\begin{remark} \label{equal_weight}
If the underlying graph of the network topology is scalar-weighted, i.e., $L_1=\cdots =L_k=\bar L$, structural controllability of this kind of networked systems falls into the SISO case with subsystem output matrix being $c_1+\cdots+c_k$ and subsystem input matrix being $b$, noting that $A_{\rm sys}=I\otimes A-\sum \nolimits_{k=1}^r L_k\otimes(bc_k)=I\otimes A- \bar L\otimes(b(c_1+\cdots c_k))$.
 It is easy to see that observability of $(A,c_1+...+c_r)$ implies that $(A, [c_1^{\intercal},...,c_r^{\intercal}]^{\intercal})$ is observable, while the converse is not necessarily true. Compared with the former case, this verifies the intuition that allowing vector-valued weights makes the conditions for structural controllability less restrictive than that of scalar-valued ones.
\end{remark}

We are now extending Theorem \ref{theorem_simo} to the case where ${\cal G}_{\rm sys}$ contains both directed and undirected edges. That is to say, not all off-diagonal entries of the Laplacian matrices $L_i|_{i=1}^r$ need to be equal to their symmetric ones, {and a nonzero entry of $L_i|_{i=1}^r$ may even have a symmetric entry which is fixed zero}.  A pair of symmetrically equal entries of $L_i$ correspond to an undirected edge, whiles an entry not equaling its symmetrical one corresponds to a directed edge, for $i=1,...,r$. We call such constraint as the {\emph{semi-symmetric constraint}} with a little abuse of terminology. Semi-symmetric constraints may emerge, such as, in a networked system where both bidirectional interactions and unidirectional interactions exist. {Semi-symmetric constrained topologies cover both the directed topology and undirected one, and are more general than them. }



Given a semi-constrained topology ${\cal G}_{\rm sys}$, let $\bar {\cal G}_{\rm sys}$ be defined in the same way as before in this section. We say  $\bar {\cal G}_{\rm sys}$ is globally input-reachable, if for each vertex $i\in {\cal V}_{\rm sys}$, there is a path consisting of either directed, undirected, or both directed and undirected edges beginning from a $u\in {\cal U}$ ending at $i$.


\begin{theorem} \label{theorem_semi_symmetric}
Consider the networked system (\ref{sub_dynamic})-(\ref{sub_interaction}) with semi-symmetric constrained topology ${\cal G}_{\rm sys}$. Assume that $|{\cal I}_u|<N$. The system is structurally controllable, if and only if $(A,b)$ is controllable, $(A, [c_1^{\intercal},...,c_r^{\intercal}]^{\intercal})$ is observable, and $\bar {\cal G}_{\rm sys}$ is globally input-reachable.
\end{theorem}
\vspace*{-1.2em}
\section{Analysis} \label{proofs}
This section gives the proofs of Theorem \ref{theorem_simo} and Theorem \ref{theorem_semi_symmetric}.

{\bf{\emph{ Proof of Necessary part of Theorem \ref{theorem_simo}}}}:  The first part of Condition 1) follows directly from Theorem 1 of \cite{Y_Zhang_2016}. The second part of Condition 1) is a direct derivation of Theorem 4 of \cite{L.Wa2016Controllability}. For Condition 2), suppose there are in total $p$ vertices that are not input-reachable in $\bar {\cal G}_{\rm sys}$. As ${\cal G}_{\rm sys}$ is undirected, by suitable reordering of vertices, $[A_{\rm sys}, B_{\rm sys}]$ has the following form
$$\left[
  \begin{array}{ccc}
    A_{11} & 0 & 0\\
    0 & A_{12} & \Delta_2\otimes b\\
  \end{array}
\right],$$
where $A_{11}$, $A_{22}$ and $\Delta_2$ are of dimensions of $np\times np$, $(N-p)n \times (N-p)n$, and $(N-p)n\times N$, respectively. This indicates that $(A_{\rm sys}, B_{\rm sys})$ cannot be controllable by the PBH test. \hfill $\square$

Our proof for sufficient part of Theorem \ref{theorem_simo} is based on the linear parameterization \cite{Morse_1976}. Consider a linear-parameterized pair $(A,B)$ modeled as
\begin{equation} \label{Linear Parameterization} A=A_0+\sum\nolimits_{i=1}^k g_is_ih_{1i}^{\intercal}, B=B_0+\sum\nolimits_{i=1}^k g_is_ih^{\intercal}_{2i}.\end{equation}
where $A_0\in {\mathbb R}^{n\times n}$, $B_0\in {\mathbb R}^{n\times m}$,  $g_i,h_{1i}\in {\mathbb R}^n$, $h_{2i}\in {\mathbb R}^m$, and $s_1,...,s_k$ are real free parameters. The pair $(A,B)$ in (\ref{Linear Parameterization}) is said to be structurally controllable, if there exists one set of values for $s_1,...,s_k$, such that the associated system is controllable.

\begin{definition} Given an $n\times n$ matrix $H$ and an $n\times m$ matrix $P$, the auxiliary digraph associated with $(H,P)$ is denoted by ${\cal G}_{\rm aux}(H,P)$, which is defined as the digraph $({\cal V}_{H}\cup {\cal V}_P, {\cal E}_{HH}\cup {\cal E}_{HP})$, where ${\cal V}_{H}=\{v_1,...,v_n\}$, ${\cal V}_P=\{u_1,...,u_m\}$, and ${\cal E}_{HH}=\{(v_i,v_j): H_{ji}\ne 0\}$, ${\cal E}_{HP}=\{(u_i,v_j): {P}_{ji}\ne 0\}$.\footnote{Here, $H_{ji}\ne 0$ means that $H_{ji}$ is not identically zero, so is with $P_{ji}\ne 0$.} A vertex $v\in {\cal V}_H$ is input-reachable in ${\cal G}_{\rm aux}(H,P)$, if there is a path from one vertex in ${\cal V}_P$ ending at $v$. A cycle of ${\cal G}_{\rm aux}(H,P)$ is said to be input-reachable, if there is at least one vertex in this cycle that is input-reachable.
\end{definition}

Define two transfer function matrices (TFMs) as $G_{zv}(\lambda)=[h_{11},...,h_{1k}]^{\intercal}(\lambda I-A_0)^{-1}[g_1,...,g_k]$, $G_{zu}(\lambda)=[h_{11},...,h_{1k}]^{\intercal}(\lambda I-A_0)^{-1}B_0+[h_{21},...,h_{2k}]^{\intercal}$. The following lemma gives necessary and sufficient conditions for $(A,B)$ in (\ref{Linear Parameterization}) to be structurally controllable.

\begin{lemma}[\cite{Morse_1976}, \cite{zhang2019structural}] \label{Theorem of Linear Parameterization}
The pair $(A,B)$ in (\ref{Linear Parameterization}) is structurally controllable, if and only if

1) Every cycle is input-reachable in ${\mathcal G}_{\rm aux}(G_{zv}(\lambda),G_{zu}(\lambda))$~;

2) For each $\lambda_0\in \sigma(A_0)$, ${\rm grank}[\lambda_0 I- A_0- \sum\nolimits_{i=1}^k g_is_ih_{1i}^{\intercal}, B_0+\sum\nolimits_{i=1}^k g_is_ih^{\intercal}_{2i}]=n$, where ${\rm grank}(\bullet)$ is the maximum rank a matrix can achieve as the function of its free parameters.
\end{lemma}

Due to dependencies among nonzero entries of $L_i$, before utilizing Lemma \ref{Theorem of Linear Parameterization}, we need to diagonalize $L_i$, $i=1,...,r$.
To this end, first arbitrarily assign an orientation to each undirected edge of ${\cal G}_{\rm sys}$, and let ${\cal E}_{\rm sys}=\{e_1,\cdots,e_{|{\cal E}_{\rm sys}|}\}$. Then construct the $|{\cal E}_{\rm sys}|\times |{\cal V}_{\rm sys}|$ incidence matrix $K_I$ as follows: $[K_I]_{ij}=1$ (, $[K_I]_{ij}=-1$) if vertex $j$ is the starting vertex (ending vertex) of $e_i$, and the remaining entries are zero. Then, define a $|{\cal V}_{\rm sys}|\times |{\cal E}_{\rm sys}|$ matrix $K$ as $K=-K_I^{\intercal}$.
It can be validated that $L_i=-K\Lambda_iK_I$, where $\Lambda_i$ is a diagonal matrix whose $j$th diagonal equals the weight of $e_j$ associated with $L_i$, $j=1,...,|{\cal E}_{\rm sys}|$.
We then have the following linear parameterization of $(A_{\rm sys}, B_{\rm sys})$
\begin{equation} \label{linear_pa2}
\begin{split}
[A_{\rm sys},B_{\rm sys}]=[I\otimes A, \Delta \otimes b]+[K\otimes b,...,K\otimes b]\\
{\bf diag}\{\Lambda_1,...,\Lambda_r\}[[K_I^{\intercal}\otimes c_1^{\intercal},...,K_I^{\intercal}\otimes c_r^{\intercal}]^{\intercal},0].
\end{split}
\end{equation}
Regarding the linear parameterization (\ref{linear_pa2}), direct algebraic manipulations show that the associated TFMs are respectively {\small{
\begin{equation} \label{g_zv} \begin{aligned}
{G_{{{zv}}}}(\lambda ) &= \left[ {\begin{array}{*{20}{c}}
{{K_I} \otimes {c_1}}\\
 \vdots \\
{{K_I} \otimes {c_r}}
\end{array}} \right]{(\lambda I - I \otimes A)^{ - 1}}[K \otimes b,\cdots, K \otimes b]\\
 &= {\bf 1}_{1\times r}\otimes \left[ {\begin{array}{*{20}{c}}
{({K_I}K) \otimes [{c_1}{{(\lambda I - A)}^{ - 1}}b)]}\\
 \vdots \\
{({K_I}K) \otimes [{c_r}{{(\lambda I - A)}^{ - 1}}b)]}
\end{array}} \right], \\
{G_{{{zu}}}}(\lambda) &=\left[
    \begin{array}{c}
      (K_I\Delta)\otimes (c_1(\lambda I-A)^{-1}b) \\
      \vdots \\
      (K_I\Delta)\otimes (c_r(\lambda I-A)^{-1}b)
    \end{array}
  \right].
\end{aligned}\end{equation}}}To proceed with our proof, we need the following lemma.
\begin{lemma}[Lemma 8 of \cite{zhang2019Structural_arxiv}] \label{lemma_equal_graph}
Given four matrices $H,P,G$ and $\Lambda$, suppose the following conditions hold: 1) $H, P$ and $G$ are of the dimensions $k\times n$, $k\times m$ and $n\times k$ respectively; 2) Whenever there exists one $l\in \{1,...,k\}$ such that $G_{il}\ne 0$ and $H_{li}\ne 0$ (resp. $G_{il}\ne 0$ and $P_{li}\ne 0$), it implies that $[GH]_{ij}\ne 0$ (resp. $[GP]_{ij}\ne 0$); 3) $\Lambda$ is an $n\times n$ diagonal matrix whose diagonal entries are free parameters. Then, every cycle is input-reachable in ${\cal G}_{\rm aux}(GH,GP)$, if and only if such property holds in ${\cal G}_{\rm aux}(H\Lambda G,P)$.
\end{lemma}

\begin{proposition} \label{reachable} If $(A,b)$ is controllable and the network topology is globally input-reachable, then every cycle of ${\cal G}_{\rm aux}(G_{zv}(\lambda),G_{zu}(\lambda))$ is input-reachable.
\end{proposition}

\begin{proof}
Since $(A,b)$ is controllable, $(I,A,b)$ is output-controllable (see \citep[Section 9.6]{Modern_Control_Ogata}). According to \cite{Modern_Control_Ogata}, this requires that the rows of $(\lambda I-A)^{-1}b$ are linearly independent in the field of complex values. That is, there cannot exist $x\in {\mathbb C}^{n}\backslash \{0\}$ making $x^{\intercal}(\lambda I-A)^{-1}b\equiv 0$.  As $c_i\ne 0$, $c_i(\lambda I-A)^{-1}b\ne 0$. Hence, $[G_{zv}(\lambda), G_{zu}(\lambda)]$ has the same sparsity pattern~as \begin{equation}\label{global} [{\bf 1}_{r\times r}\otimes (K_IK), {\bf 1}_{r\times 1}\otimes (K_I\Delta)].\end{equation}
Let $\widetilde {\cal G}_{\rm {aux}}$ be the auxiliary digraph associated with (\ref{global}).
Define matrices $G\doteq {\bf diag}\{K_I|_{i=1}^r\}$, $H\doteq {\bf 1}_{r\times r}\otimes K$, $P\doteq {\bf 1}_{r\times 1}\otimes \Delta$. Then,
 (\ref{global}) can be expressed as $[GH,GP]$.
From the construction of $K_I$ and $K$, one has that
\[{[K_IK]_{ij}} = \left\{ \begin{array}{l}
 {[K_I]_{il_1}K_{l_1i}+[K_I]_{il_2}K_{l_2i}=-2}, {i = j, e_{i}=(l_1,l_2)}\\
 {[K_I]_{il_1}K_{l_1j}=-1}, {i\ne j, l_1= V(e_i)\cap V(e_j)\ne \emptyset}\\
0,{i \ne j, V(e_i)\cap V(e_j)= \emptyset}
\end{array}, \right.\]where $V(\bullet)$ denotes the vertex set.
Hence, it can be validated that, Condition 2) of Lemma \ref{lemma_equal_graph} holds with respect to $(G,H,P)$.

Let $L$ be a Laplacian matrix associated with ${\cal G}_{\rm sys}$.
Using Lemma \ref{lemma_equal_graph} on $(G,H,P)$,  one will obtain the following matrix
\[[({\bf 1}_{r\times r}\otimes K){\bf diag}\{\Lambda_iK_I|_{i=1}^r\}, {\bf 1}_{r\times 1}\otimes \Delta]\]
which has the same sparsity pattern as \begin{equation} \label{same_sparsity} [-({\bf 1}_{r\times r}\otimes L), {\bf 1}_{r\times 1}\otimes \Delta] ,\end{equation}where $\Lambda_i$ is defined just before (\ref{linear_pa2}).
Denote by $\hat {\cal G}_{\rm {aux}}$ the auxiliary digraph associated with (\ref{same_sparsity}).  According to Lemma \ref{lemma_equal_graph}, every cycle in $\widetilde {\cal G}_{\rm {aux}}$ is input-reachable, if and only if such property holds in $\hat {\cal G}_{\rm {aux}}$.  By zeroing out the off-diagonal blocks of the left submatrix of (\ref{same_sparsity}), noting that $\bar {\cal G}_{\rm sys}$ is the auxiliary digraph associated with $[-L, \Delta]$ by eliminating all self-loops,  the global reachability of $\bar {\cal G}_{\rm sys}$ indicates
that, every vertex is input-reachable in $\hat {\cal G}_{\rm aux}$. By Lemma \ref{lemma_equal_graph}, the proposed proposition is proved.
\end{proof}


\begin{proposition}\label{rank_proposition}
Consider the networked system (\ref{sub_dynamic})-(\ref{sub_interaction}). If $\bar {\cal G}_{\rm sys}$ is globally input-reachable, whiles $(A,b)$ is controllable and $(A,[c_1^{\intercal},...,c_r^{\intercal}]^{\intercal})$ is observable. Then, for each $\lambda_i\in \sigma(A)$, ${\rm grank}[\lambda_i I- A_{\rm sys}, B_{\rm sys}]= Nn$.
\end{proposition}
\vspace*{-0.2em}
To prove the above proposition, we need the following lemma.
\vspace*{-0.1em}
\begin{lemma} \label{important}
Given matrices $A\in {\mathbb R}^{n\times n}$, $b\in {\mathbb R}^{n}$, $C\in {\mathbb R}^{r\times n}$, suppose that $(A,b)$ is controllable and $(A,C)$ is observable. Let $\Lambda=[s_1,\cdots,s_r]$ be a $1\times r$ matrix whose entries are all free parameters $s_i|_{i=1}^r$. Then, for arbitrary $Q_0\in {\mathbb C}^{r\times 1}$,
$$
{{\rm grank}\left[
             \begin{array}{cc}
               \lambda_iI-A & b\Lambda \\
               -C & I_r-Q_0\Lambda \\
             \end{array}
           \right]=n+r}
$$
holds for each $\lambda_i\in \sigma(A)$.
\end{lemma}
\begin{proof}
Define $S\doteq {\bf diag}\{s_i|_{i=1}^r\}$ and $M_0=-{\bf 1}_{1\times r}\otimes Q_0$. Let $\bar s_i\doteq s_i^{-1}$, $i=1,...,r$. We then have
$$
\left[\!
             \begin{array}{cc}
               \lambda_iI-A & b\Lambda \\
               -C & I_r-Q_0\Lambda \\
             \end{array}
           \!\right]\!=\! \underbrace{\left[\begin{array}{cc}
               \lambda_iI-A & {\bf 1}_{1\times r}\otimes b  \\
               -C & S^{-1}+M_0 \\
             \end{array}\right]}_{\doteq F(\lambda_i)}\!\left[\!
                                                        \begin{array}{cc}
                                                          I_n &  \\
                                                           & S\\
                                                        \end{array}
                                                      \!\right]\!.
$$It follows that, if $F(\lambda_i)$ is full column generic rank for each $\lambda_i\in \sigma(A)$, then the proposed statement is proved. Define matrix
\begin{equation}\label{gamma}{\small \Gamma=\left[
          \begin{array}{cccc}
            1 & 0  & \cdots & 0  \\
            -1 & 1 & \cdots  & 0  \\
            \vdots &   & \ddots &  \\

            0 &   & \cdots &  1 \\
            0 &   &   \cdots & -1 \\
          \end{array}
        \right]\in {\mathbb R}^{r\times (r-1)}.}
\end{equation}
As $(A,b)$ is controllable, it can be validated that, for any $\lambda_i\in \sigma(A)$, a basis matrix spanning the null space of $[\lambda_i I-A, {\bf 1}_{1\times r}\otimes b]$ can be
${\bf diag}\{x_i,\Gamma\}$, where $x_i$ is the eigenvector of $A$ associated with $\lambda_i$. Notice that $x_i$ is unique if scaled by certain scalars.
Then, by the property of null space (see \cite{R.A.1991Topics}),  $F(\lambda_i)$ is of full (column) generic rank, if and only if
$$[-C, S^{-1}+M_0]{\bf diag}\{x_i,\Gamma\}=[-Cx_i,S^{-1}\Gamma+M_0\Gamma]$$ is of full (column) generic rank.
As $(A,C)$ is observable, there exists at least one $j\in\{1,...,r\}$, such that $[Cx_i]_j\ne 0$. Consider $\det [-Cx_i,S^{-1}\Gamma+M_0\Gamma]$. According to the structure specificity of $\Gamma$ in (\ref{gamma}), there exists one and only one nominal $\prod\nolimits_{k=1,k\ne j}^r \bar s_k$ in $\det [-Cx_i,S^{-1}\Gamma+M_0\Gamma]$, whose coefficient is $[Cx_i]_j$ (ignoring the corresponding sign). This can be validated by the definition of determinant, i.e., determinant of an $n\times n$ matrix is the sum of signed products of all $n$ entries, each chosen from different rows and different columns. Notice that $\prod\nolimits_{k=1,k\ne j}^{k=r}\bar s_k$ has the maximum degree $r-1$ such that arbitrary constant $M_0$ cannot violate the aforementioned fact.  Hence, $[-Cx_i,S^{-1}\Gamma+M_0\Gamma]$ is of full generic rank, proving Lemma \ref{important}.
\end{proof}
\vspace*{-0.3em}

{\bf{\emph{{Proof of Proposition \ref{rank_proposition}:}}}} We will use mathematical induction. Since $\bar {\cal G}_{\rm sys}$ is globally input-reachable, first assume that $\bar {\cal G}_{\rm {sys}}$ has one spanning tree rooted at $u_1$. Denote this spanning tree by $\cal T$, and vertices $u_1,1,...,N$ are in lexicographic order in the sense that the parent of vertex $i$ is among vertices $i-1,...,1,u_1$ in $\cal T$, $i=1,...,N$. Moreover, let ${\cal T}_i$ be the subgraph of $\cal T$ induced by vertices $1,...,i$, $i=1,...,N$, and $K_{Ii}$ be the incidence matrix associated with ${\cal T}_i$, which is defined similarly to $K_I$ of $\cal G_{\rm sys}$. In this sense, $K_{Ii}$ can be recursively constructed as
\begin{equation}\label{Krelation}K_{I,i+1}=\left[
  \begin{array}{cc}
    K_{Ii} & 0 \\
    (e^{[i]}_{(i+1)^*})^\intercal & -1 \\
  \end{array}
\right],\end{equation}
where $(i+1)^*$ is the parent of vertex $i+1$ in $\cal T$, $i=0,...,N-1$, and $K_{I0}$ is empty. Let $K_{i}=-K_{Ii}^{\intercal}$.  Let weights of edges not in $\cal T$ be zero. Remember that $W_{ii^*}$ is the vector-valued weight of the edge connecting vertex $i$ and its parent $i^*$.
Let $S_i={\bf diag}\{W_{jj^*}|_{j=2}^{i+1}\}$, $i=1,...,N-1$, and $S_0$ be empty, i.e., $S_i$ stores all  weights of edges in ${\cal T}_{i+1}$. Then, it can be validated that \begin{equation}\label{lumpA} A_{\rm sys}=I_N\otimes A-(K^{\intercal}_{IN}\otimes b)S_{N-1}(K_{IN}\otimes C)\end{equation}
We will prove by induction that, for each $\lambda_0\in {\sigma(A)}$
\begin{equation}\label{induction}[I_i\otimes A- (K^{\intercal}_{Ii}\otimes b)S_{i-1}(K_{Ii}\otimes C)-\lambda_0 I, e_1^{[i]}\otimes b]\end{equation}
is of full row generic rank for $i=1,...,N$. Since $(A,b)$ is controllable, the base case where $i=1$ is obviously true. Now suppose that (\ref{induction}) is of full row generic rank for some $i$ between $1$ and $N-1$.  Rewrite (\ref{induction}) as
\begin{equation}\begin{split}[I_i\otimes A- (K^{\intercal}_{Ii}\otimes b)S_{i-1}(K_{Ii}\otimes C)-\lambda_0 I, e_1^{[i]}\otimes b]\\
=[I_i\otimes A-\lambda_0 I, e_1^{[i]}\otimes b]+(K_{i}\otimes b)S_{i-1}[K_{Ii}\otimes C, 0]
\end{split}.
\end{equation}
Using Schur complement \cite{R.A.1991Topics} on the above formula, we have that (\ref{induction}) is of full row generic rank, if and only if
\[\Psi_i\doteq\begin{pmat}[{.|}]
   I_i\otimes A-\lambda_0 I & e_1^{[i]}\otimes b & (K_{Ii}^{\intercal}\otimes b)S_{i-1} \cr\-
    K_{Ii}\otimes C & 0 & I \cr
\end{pmat}\]is so. Substituting (\ref{Krelation}) into $\Psi_{i+1}$ and after some elementary permutations, $\Psi_{i+1}$ is of full row generic rank,  if and only if
\begin{equation}\label{bigmatrix}{\footnotesize\begin{pmat}[{..|.}]
   I_i\otimes A-\lambda_0 I & e_1^{[i]}\otimes b & (K^{\intercal}_{Ii}\otimes b)S_{i-1} & 0 & e^{[i]}_{(i+1)^*}\otimes b \Lambda_{i+1} \cr
    K_{Ii}\otimes C & 0 & I & 0 & 0 \cr\-
    0 & 0 & 0 & A-\lambda_0 I & -b\Lambda_{i+1} \cr
    (e^{[i]}_{(i+1)^*})^{\intercal}\otimes C & 0 & 0 & -C & I_r\cr
\end{pmat}}\end{equation} is of full row generic rank, where $\Lambda_{i+1}\doteq W_{(i+1)(i+1)^*}$ for notation simplicity.
Let $S_{i-1}$ take some value such that $\Psi_i$ is of full row rank. Then for arbitrary $\Lambda_{i+1}\in {\mathbb R}^{1\times r}$, {\footnotesize $\Psi_i\bar \Psi_i=\left[
                                                                                                        \begin{array}{c}
                                                                                                          e^{[i]}_{(i+1)^*}\otimes b\Lambda_{i+1} \\
                                                                                                          0 \\
                                                                                                        \end{array}
                                                                                                      \right]
 $}, where {\footnotesize$\bar \Psi_i\doteq \Psi_i^{\dag}\left[
                                                                                                        \begin{array}{c}
                                                                                                          e^{[i]}_{(i+1)^*}\otimes b\Lambda_{i+1} \\
                                                                                                          0 \\
                                                                                                        \end{array}
                                                                                                      \right]$}, and $(\bullet)^{\dag}$ denotes the Moore-Penrose inverse. Hence, post-multiplying  {\footnotesize $\small{\left[
                                                                                                        \begin{array}{ccc}
                                                                                                          I & 0 & - \bar \Psi_i\\
                                                                                                          0 & I_n & 0 \\
                                                                                                          0 & 0 & I_r \\
                                                                                                        \end{array}
                                                                                                      \right]}$} to (\ref{bigmatrix}), one will obtain  { $$\Pi \doteq \left[
                   \begin{array}{cc}
                     \Psi_i & 0 \\
                     \Pi_{21} & \Pi_{22} \\
                   \end{array}
                 \right],$$}where {\footnotesize $\Pi_{21}\doteq \left[
                                               \begin{array}{ccc}
                                                 0 & 0 & 0 \\
                                                (e^{[i]}_{(i+1)^*})^{\intercal}\otimes C  & 0 & 0 \\
                                               \end{array}
                                             \right]
                 $, $\Pi_{22}\doteq \left[
                                  \begin{array}{cc}
                                    A-\lambda_0 I & -b\Lambda_{i+1} \\
                                    -C &  I_r-Q_0\Lambda_{i+1}\\
                                  \end{array}
                                \right]
                 $}, with constant matrix $Q_0\in {\mathbb C}^{r\times 1}$ satisfying  {\footnotesize $[(e^{[i]}_{(i+1)^*})^{\intercal}\otimes C,0,0]\bar \Psi_i=Q_0\Lambda_{i+1}$}.
By Lemma \ref{important}, $\Pi_{22}$ is of full generic rank. Hence, $\Pi$ is of full row generic rank, which means that (\ref{bigmatrix}) is so, too. Thus, replacing $i$ with $i+1$, (\ref{induction}) is of full row generic rank. Inducing from $i=1$ to $i=N$, the proposed statement is proved.

The case that $\bar {\cal G}_{\rm {sys}}$ can be decomposed into more than one vertex-disjoint spanning trees (all rooted at $\cal U$) follows immediately from the former case.
\hfill $\square$

{\bf{\emph{Proof of sufficient part of Theorem \ref{theorem_simo}}}}: By Lemma \ref{Theorem of Linear Parameterization}, Propositions \ref{reachable} and \ref{rank_proposition} assure the sufficiency of conditions in Theorem \ref{theorem_simo} for structural controllability. \hfill $\square$
\begin{spacing}{1.08}
{\bf{\emph{Proof of Theorem \ref{theorem_semi_symmetric}}}}:
 To handle with semi-symmetric constraints, we shall modify the diagonalization of $L_i|_{i=1}^r$.
To this end, first for each undirected edge of ${\cal G}_{\rm sys}$, arbitrarily assign an orientation, whereas each directed edge remains unchanged. Then the incidence matrix $K_I$ is defined as a $|{\cal E}_{\rm sys}|\times |{\cal V}_{\rm sys}|$ matrix (each undirected edge counts one for $|{\cal E}_{\rm sys}|$), where $[K_I]_{ij}=1$ (resp. $[K_I]_{ij}=-1$) if vertex $j$ is the starting vertex (resp. ending vertex) of the $i$th edge. Moreover, $K$ is a $|{\cal V}_{\rm sys}|\times |{\cal E}_{\rm sys}|$ matrix defined as follows{\small
\[{K_{ij}} = \left\{ \begin{array}{l}
1,  {\text{ if}} \quad {[{K_I}]_{ji}} =  - 1\\
-1,  {\text{ if the jth edge is undirected and}} \quad {[{K_I}]_{ji}} =  1\\
0, {\rm otherwise}
\end{array} \right.\]}Afterwards, letting $\Lambda_i$ be the diagonal matrix whose $j$th diagonal entry is the weight of the $j$th edge of ${\cal G}_{\rm sys}$ associated with $L_i$, $j=1,...,|{\cal E}_{\rm sys}|$, one has
$L_i=-K\Lambda_i K_I$. Based on such
diagonalization, the rest of the proof follows similar arguments
to that of Theorem \ref{theorem_simo}. Details are omitted due to their similarities.
\hfill $\square$
\end{spacing}
\vspace*{-1em}
{
\section{Extension with Matrix-Weighted Edges}
This section extends Theorem \ref{theorem_simo} to the case with arbitrarily dimensional matrix-weighted edges. We will give a sufficient condition for structural controllability.
Let us modify subsystem dynamics in (\ref{sub_dynamic})-(\ref{sub_interaction}) into multi-input-multi-output (MIMO)
\begin{equation}\label{matrix_sys1} x_i(t)=Ax_i(t)+Bv_i(t),\end{equation}
\begin{equation}\label{matrix_sys2} v_i(t)=\sum \nolimits_{j=1}^NW_{ij}C(x_j(t)-x_i(t))+\delta_i u_i(t),\end{equation}
where $A\in {\mathbb R}^{n\times n}$, $B\in {\mathbb R}^{n\times p}$ and $C\in {\mathbb R}^{r\times n}$. Matrix $W_{ij}\in {\mathbb R}^{p\times r}$ is a {\emph{matrix-valued weight}} of edge $(j,i)\in {\mathcal E}_{\rm sys}$ with symmetric constraint $W_{ij}=W_{ji}$.  Moreover,  $W_{ij}=0$ if $(j,i)\notin {\mathcal E}_{\rm sys}$.
The lumped representation of the system becomes
\begin{equation}\label{matrix_lump_equation}A_{\rm sys}=I_N\otimes A- (I_N\otimes B)L_m(I_N\otimes C), B_{\rm sys}= \Delta\otimes (B),\end{equation}
where $\Delta$ is defined in the same way as Section \ref{model_description}, and $L_m$ is
{\small$$L_m=\left[
           \begin{array}{cccc}
             \sum \nolimits_{j=1}^N W_{1j} & -W_{12} & \cdots & -W_{1N} \\
             -W_{21} & \sum \nolimits_{j=1}^N W_{2j} & \cdots & -W_{2N} \\
             \vdots & \vdots & \ddots & \vdots \\
             -W_{N1} & -W_{N2} & \cdots & \sum \nolimits_{j=1}^N W_{Nj} \\
           \end{array}
         \right]\in {\mathbb R}^{Np\times Nr},
$$}which could be called a matrix-weighted Laplacian for ${\cal G}_{\rm sys}$. Examples of the above networked system include coupled mass-spring systems and coupled electrical oscillators; see \cite{tuna2016synchronization}.

To given a linear parameterization of (\ref{matrix_lump_equation}),  introduce two matrices $T\doteq I_p\otimes {\bf 1}_{1\times r}$ and $Q\doteq {\bf 1}_{p\times 1}\otimes I_r$. For each $(i,j)\in {\cal E}_{\rm sys}$, let $\Lambda_{ij}$ be a $pq\times pq$ dimensional diagonal matrix whose diagonal entries consist of all entries of $W_{ij}$,  i.e., $\Lambda_{ij}={\bf diag}\{[W_{ij}]_{11},[W_{ij}]_{12},\cdots ,[W_{ij}]_{1r},\cdots,[W_{ij}]_{pr}\}$. Then, we have
$$W_{ij}=T\Lambda_{ij}Q.$$
Moreover, let $K_I$ be the incidence matrix of ${\cal G}_{\rm sys}$, whose definition is given in Section \ref{proofs}, and $K=-K_I^{\intercal}$. Then, it can be validated that,  (\ref{matrix_lump_equation}) has the following linear parameterization,
\begin{equation}\label{matrix_weight}
\begin{array}{l}
A_{\rm sys}=I_N\otimes A-(I_N\otimes B)(K\otimes T){\bf diag}\{\Lambda_{ij}|_{(i,j)\in {\cal E}_{\rm sys}}\} \\
(K_I\otimes Q)(I_N\otimes C), \quad B_{\rm sys}=\Delta \otimes B, \end{array}
\end{equation}where the diagonal entries of ${\bf diag}\{\Lambda_{ij}|_{(i,j)\in {\cal E}_{\rm sys}}\}$ are placed in the order consistent with the incidence matrix $K_I$.

Regarding the linear parameterization (\ref{matrix_weight}), by some algebraic manipulations, the associated TFMs are respectively
\[\begin{split} G^{m}_{zv}(\lambda)\!\!&=\!\!(K_I\otimes Q)(I_N\otimes C)(I_N\otimes(\lambda I-A)^{-1})(I_N\otimes B)(K\otimes T)\\
\!\!&=(K_IK)\otimes(QC(\lambda I-A)^{-1}BP)\\
G^{m}_{zu}(\lambda)\!\!&=(K_I\otimes Q)(I_N\otimes C)(I_N\otimes(\lambda I-A)^{-1})(\Delta \otimes B)\\
&=(K_I\Delta)\otimes(QC(\lambda I-A)^{-1}B)\\
\end{split}\]
\begin{definition}[fixed mode, \cite{anderson1981algebraic}] Given a triple $(A,B,C)\in {\mathbb R}^{n\times n}\times {\mathbb R}^{n\times p}\times {\mathbb R}^{r\times n}$, $(A,B,C)$ is said to have no fixed mode, if $\bigcap \nolimits_{F\in {\mathbb R}^{p\times r}}\sigma(A+BFC)=\emptyset$.
\end{definition}
\begin{proposition}\label{pro_1}  Suppose that for the networked system (\ref{matrix_sys1})-(\ref{matrix_sys2}), $(A,B,C)$ has no fixed mode. If the network topology is globally input-reachable, then every cycle of ${\cal G}_{\rm aux}(G^{m}_{zv}(\lambda),G^{m}_{zu}(\lambda))$ is input-reachable.
\end{proposition}
\begin{proof} See the appendix. 
\end{proof}

\begin{proposition}\label{pro_2} Suppose that for the networked system (\ref{matrix_sys1})-(\ref{matrix_sys2}), $(A,B,C)$ has no fixed mode. If the network topology is globally input-reachable, then ${\rm grank}[\lambda_i I- A_{\rm sys}, B_{\rm sys}]= Nn$ for each $\lambda_i \in \sigma(A)$.
\end{proposition}
\begin{proof} Observe that $A_{\rm sys}$ in (\ref{matrix_lump_equation}) can be rewritten as \begin{equation} \label{matrixB} A_{\rm sys}=I_N\otimes A-(K_I^{\intercal}\otimes B){\bf diag}\{W_{ij}|_{(i,j)\in {\cal E}_{\rm sys}}\}(K_I\otimes C),\end{equation}
where the diagonal entries of ${\bf diag}\{W_{ij}|_{(i,j)\in {\cal E}_{\rm sys}}\}$ are in the order consistent with $K_I$. Notice that (\ref{matrixB}) has the same form as (\ref{lumpA}). This means that, if we replace the vector $b\in {\mathbb R}^n$ in Lemma \ref{important} with a matrix $B\in {\mathbb R}^{n
\times p}$ and show that the associated implications in that lemma still hold under the proposed condition in Proposition {\ref{pro_2}}, then we could prove Proposition {\ref{pro_2}} in the same line as that of Proposition \ref{rank_proposition}.  For this purpose, we will prove that, for any $\lambda_0\in \sigma(A)$, if $\bigcap \nolimits_{F\in {\mathbb R}^{p\times r}}\sigma(A+BFC)=\emptyset$, then for arbitrary $Q_0\in {\mathbb R}^{r\times p}$, there exists a $F_0\in {\mathbb R}^{p\times r}$, such that matrix
\[M(F_0)=\left[
  \begin{array}{cc}
    A-\lambda_0I & -BF_0 \\
    -C & I-Q_0F_0 \\
  \end{array}
\right]\]has full row rank.
In fact, if $\bigcap \nolimits_{F\in {\mathbb R}^{p\times r}}\sigma(A+BFC)=\emptyset$, there exists $W\in {\mathbb R}^{p\times r}$, such that
$A-\lambda_0I-BWC$ and $I+WQ_0$ are simultaneously invertible. This can be justified by the following analysis. Suppose that a matrix $W_0$ exists such that $A-\lambda_0I-BW_0C$ is invertible. Then, it is an easy manner to see that the set $\Delta_0\doteq \{\Delta W\in {\mathbb R}^{p\times r}: A-\lambda_0I-B(W_0+\Delta W)C \quad {\text {is}} \quad {\text {singluar}}\}$ has zero Lebesgue measure in ${\mathbb R}^{p\times r}$. On the other hand, the set $\Delta_1\doteq \{\Delta W \in  {\mathbb R}^{p\times r}: I+(W_0+\Delta W)Q_0 \quad {\text {is}} \quad {\text {singluar}}\}$ also has zero Lebesgue measure in ${\mathbb R}^{p\times r}$, noting that when $\Delta W=-W_0$, $I+(W_0+\Delta W)Q_0=I$ is invertible. Hence, each element $\Delta W\in {\mathbb R}^{p\times r}\setminus (\Delta_0 \bigcup \Delta_1)$ making $A-\lambda_0I-BWC$ and $I+WQ_0$ simultaneously invertible, with $W=W_0+\Delta W$. Let $F_0=(I+WQ_0)^{-1}W$. It can be validated that, $$A-\lambda_0 I- BF_0(I-Q_0F_0)^{-1}C=A-\lambda_0 I-BWC.$$
That means, $A-\lambda_0 I- BF_0(I-Q_0F_0)^{-1}C$ is invertible, which according to the property of Schur complement, indicates that $M(F_0)$ is invertible. Based on the above, the proposed statement follows similar arguments to the proof of Proposition \ref{rank_proposition}.
\end{proof}

By Lemma \ref{Theorem of Linear Parameterization}, the following theorem follows immediately from Propositions \ref{pro_1}-\ref{pro_2}.
\begin{theorem}\label{theorem_matrix} Consider the networked system (\ref{matrix_sys1})-(\ref{matrix_sys2}) with undirected ${\cal G}_{\rm sys}$.  Suppose that $(A,B,C)$ has no fixed mode. Then, this system is structurally controllable, if and only if the network topology $\bar {\cal G}_{\rm sys}$ is globally input-reachable.
\end{theorem}

By characterizations of fixed mode \cite{anderson1981algebraic}, it can be validated that Theorem \ref{theorem_simo} is a special case of Theorem \ref{theorem_matrix}. {\emph{However, unlike Theorem \ref{theorem_simo}, the nonexistence of fixed mode is not necessarily necessary  in the case with matrix-weighted edges.}}}

%
%

\vspace*{-0.6em}
\section{Conclusions}
In this paper, we have proved that, an undirected networked system with {diffusively coupled} identical high-order SIMO subsystems is structurally controllable, if and only if each subsystem is controllable and observable, and the network topology is globally input-reachable. It is also demonstrated that, such conditions are still necessary and sufficient when both directed and undirected edges exist in the network topology. In these results, the underlying graph of the network topology is vector-weighted. {An extension has been further  given when each subsystem is MIMO and the interaction links are matrix-weighted.}

\begin{appendix}
\begin{lemma}[Corollary 52 of \cite{Master_Kronecker}] \label{commutable}
Let $A$ and $B$ be two matrices with dimensions $m\times n$ and $p\times r$ respectively. Let $P(n_1,n_2)$ be an $n_1n_2\times n_1n_2$ permutation matrix depending {\emph{only}} on integers $n_1$ and $n_2$. Then, there exists two permutation matrices $P(m,p)$ and $P(n,r)$, such that
$P(m,p)^{\intercal}(A\otimes B)P(n,r)=B\otimes A$. 
\end{lemma}

{\bf \emph{{Proof of Proposition \ref{pro_1}}}}: Let $H(\lambda)\doteq C(\lambda I-A)^{-1}B$. Using Lemma \ref{commutable}, we know that there exist two permutation matrices $P(|\mathcal E_{\rm sys}|, pr)$ and $P(N,p)$, such that
$P(|\mathcal E_{\rm sys}|, pr)^{\intercal}G^{m}_{zv}(\lambda)P(|\mathcal E_{\rm sys}|, pr)=(QH(\lambda)P)\otimes(K_IK)$, and
$P(|\mathcal E_{\rm sys}|, pr)^{\intercal}G^{m}_{zu}(\lambda)P(N,p)=(QH(\lambda))\otimes(K_I\Delta)$.

Note that, ${\cal G}_{\rm aux}\big(P(|\mathcal E_{\rm sys}|, pr)^{\intercal}G^{m}_{zv}(\lambda)P(|\mathcal E_{\rm sys}|, pr),P(|\mathcal E_{\rm sys}|, pr)^{\intercal}\\~G^{m}_{zu}(\lambda)P(N,p) \big)$ (denoted by $\bar {\cal G}_{\rm aux}$) can be obtained from ${\cal G}_{\rm aux}\big(G^{m}_{zv}(\lambda),G^{m}_{zu}(\lambda)\big)$ by reordering the associated vertices. Hence, every vertex in
${\cal G}_{\rm aux}(G^{m}_{zv}(\lambda),G^{m}_{zu}(\lambda))$ is input-reachable, if and only if such property holds in $\bar {\cal G}_{\rm aux}$. Since we assume that $(A,B,C)$ has no fixed mode, it is easy to see that this requires that $(A,B)$ is controllable. As $c_i\ne 0$, $i=1,...,r$, we have that $c_i(\lambda I-A)^{-1}B\ne 0$ by noting that $(I,A,B)$ is output controllable. This means that, every row of $H(\lambda)$ is nonzero. Suppose that in the $i$th row of $H(\lambda)$, the $\sigma(i)$th entry is nonzero, $\sigma(i)\in\{1,\cdots,p\}$. Let $W$ be the $r\times p$ matrix with entries being zero or one, where only the $(i,\sigma(i))$th entry is one, $i=1,...,r$. On the other hand, by definitions of matrices $T$ and $Q$, after some simple algebraic manipulations, we have
$QH(\lambda)T={\bf 1}_{p\times 1}\otimes H(\lambda) \otimes {\bf 1}_{1\times r}$, $QH(\lambda)={\bf 1}_{p\times 1}\otimes H(\lambda)$.
Hence, it is easy to see that, if every vertex is input-reachable in ${\cal G}_{\rm aux}((QWT)\otimes(K_IK),(QW)\otimes(K_I\Delta))$, then such property holds in $\bar {\cal G}_{\rm aux}$, as the former is a subgraph of the latter.

Now let $G={\bf diag}\{K_I|_{i=1}^{pr}\}$, $H={\bf 1}_{p\times 1}\otimes W\otimes {\bf 1}_{1\times r}\otimes K$, $P={\bf 1}_{p\times 1}\otimes W\otimes \Delta$. Then, $[(QWT)\otimes(K_IK),(QW)\otimes(K_I\Delta)]$ can be rewritten as $[GH,GP]$. Using Lemma \ref{lemma_equal_graph} on $(G,H,P)$, one obtain the following associated matrix
$$[\underbrace{{\bf 1}_{p\times 1}\otimes W\otimes {\bf 1}_{1\times r}}_{\doteq W_1\in {\mathbb R}^{pr\times pr}}\otimes L,\underbrace{{\bf 1}_{p\times 1}\otimes W}_{\doteq W_2\in {\mathbb R}^{pr\times p}}\otimes \Delta],$$where $L$ is a Laplacian matrix associated with ${\cal G}_{\rm sys}$. 

In the digraph ${\cal G}_{\rm aux}(W_1\otimes L, W_2\otimes \Delta)$, let the vertex corresponding to the $[N(i-1)+j]$th row of $W_1\otimes L$ be denoted by vertex $(i,j)$, $i=1,...,pr$, $j=1,...,N$. Moreover, assume that there is a spanning tree in ${\cal G}_{\rm sys}$ rooted at vertex $1$ with $\delta_1=1$, and denote this tree by $\cal T$ (similar arguments could be made if ${\cal G}_{\rm sys}$ can be decomposed into more than one disjoint spanning trees). Let $Pa(j)$ denote the parent of vertex $j$ in $\cal T$, $1\le j \le N$. Since every row of $W$ is nonzero, such property holds for $W_1$ and $W_2$, as well as $W_1\otimes L$. For each vertex $(i,j)$, $1\le i \le pr$, $1\le j \le N$, according to the structure of $W_1\otimes L$,  it can be seen that the $[\sigma({i \  {\rm mod} \  r})-1]Np+Pa(j)$ -th column of $W_1\otimes L$ is nonzero, where $n_1 \ {\rm  mod} \ n_2$ takes the remainder of $n_1$ divided by $n_2$ (if the remainder is zero, then returns $n_2$). This means that vertex $(i,j)$ has an in-neighbor $\big([\sigma({i \  {\rm mod} \  r})-1]p+1, Pa(j)\big)$. Define $f:{\mathbb N}\rightarrow {\mathbb N}$ as $f(i)=[\sigma({i \  {\rm mod} \  r})-1]p+1$. Note that vertex $(i,1)$ is input-reachable as the $[N(i-1)+1]$th row of $W_2\otimes \Delta$ is nonzero, $i=1,...,pr$. On the other hand, the existence of $\cal T$ in ${\cal G}_{\rm sys}$ means that, $\underbrace{Pa(\cdots(Pa(j))}_{{\text{no more than j-1 $Pa(\cdot)$}}}\cdots)$ reaches $1$. Hence, there is a path from vertex $\big(\underbrace{f(\cdots(f(i))}_{{\text{no more than $j-1$ $f(\cdot)$}}}\cdots),1\big)$ to $(i,j)$ in ${\cal G}_{\rm aux}(W_1\otimes L, W_2\otimes \Delta)$ for any $i\in \{1,...,pr\}$, $j\in\{2,...,N\}$, leading to the input-rechability of $(i,j)$. Based on the above arguments,  this proves the input-reachability of every cycle in ${\cal G}_{\rm aux}(G^{m}_{zv}(\lambda),G^{m}_{zu}(\lambda))$ by Lemma \ref{lemma_equal_graph}, thus proving Proposition \ref{pro_1}.  \hfill $\square$
%
\end{appendix}

\bibliographystyle{elsarticle-num}
{\footnotesize
\bibliography{yuanz3}
}

\end{document}